\documentclass[11pt]{article}

\usepackage{fullpage,wrapfig}
\usepackage{amsmath, amsthm, amssymb, algorithm,thmtools,algpseudocode,enumerate,caption,framed}
\usepackage{paralist, sidecap}
\usepackage[usenames,dvipsnames,svgnames,table]{xcolor}
\definecolor{darkgreen}{rgb}{0.0,0,0.9}
\usepackage[colorlinks=true,pdfpagemode=UseNone,citecolor=Blue,linkcolor=Red,urlcolor=Black,pagebackref]{hyperref}
\usepackage{tikz}
\usepackage{authblk}

\newtheorem{theorem}{Theorem}[section]

\newtheorem{lemma}{Lemma}[section]

\newcommand{\up}{\texttt{upper}}
\newcommand{\lo}{\texttt{lower}}
\newcommand{\A}{\mathcal{A}_e}
\newcommand{\Av}{\mathcal{A}_v}
\newcommand{\I}{\mathtt{MWIS}}
\newtheorem{myremark}{Remark}{\bfseries}{\itshape}

\title{Polygon Simplification by Minimizing Convex Corners\thanks{A preliminary version of this paper appeared in proceedings of the 22nd International Computing and Combinatorics Conference (COCOON 2016)~\cite{BahooDK0MM16}. Research of Stephane Durocher and Debajyoti Mondal is supported in part by Natural Sciences and Engineering Research Council of Canada (NSERC).}}

\author[1]{Yeganeh Bahoo}
\author[1]{Stephane Durocher}
\author[2]{J. Mark Keil}
\author[2]{Debajyoti Mondal}
\author[3]{Saeed Mehrabi}
\author[4]{Sahar Mehrpour}

\affil[1]{{\small Department of Computer Science, University of Manitoba, Winnipeg, Canada.

			\texttt{\{bahoo, durocher\}@cs.umanitoba.ca}}
}
\affil[2]{{\small Department of Computer Science, University of Saskatchewan, Saskatoon, Canada.

\texttt{\{keil, dmondal\}@cs.usask.ca}}
}
\affil[3]{{\small School of Computer Science, Carleton University, Ottawa, Canada.

			\texttt{saeed.mehrabi@carleton.ca}}
}
\affil[4]{{\small School of Computing, University of Utah, Utah, USA.

			\texttt{mehrpour@cs.utah.edu}}
}

\date{}

\begin{document}

\maketitle

\begin{abstract}
Let $P$ be a polygon  with  $r>0$ reflex vertices and possibly with holes and islands. A subsuming polygon of $P$ is a polygon $P'$ such that $P \subseteq P'$, each connected component $R$ of $P$ is a subset of a distinct connected component $R'$ of $P'$, and the reflex corners of $R$ coincide with those of $R'$. A subsuming chain of $P'$ is a minimal path on the boundary of $P'$ whose two end edges coincide with two edges of $P$. Aichholzer et al. proved that every polygon $P$ has a subsuming polygon with $O(r)$ vertices, and posed an open problem to determine the computational complexity of computing subsuming polygons with the minimum number of convex vertices. 

We prove that the problem of computing an optimal subsuming polygon is NP-complete, but the complexity remains open for simple polygons (i.e., polygons without holes). Our NP-hardness result holds even when the subsuming chains are restricted to have constant length and lie on the arrangement of lines determined by the edges of the input polygon. We show that this    restriction makes the problem polynomial-time solvable for simple polygons.
\end{abstract}

%%%%%%%%%%%% NEW SECTION %%%%%%%%%%%%%%%%%
\section{Introduction}
\label{sec:introduction}
Polygon simplification is well studied in computational geometry, with numerous applications in cartographic visualization, computer graphics and data compression~\cite{cartogram,Ratschek}. Techniques for simplifying polygons and polylines have appeared in the literature in various forms. Common goals of these simplification algorithms include to preserve the shape of the polygon, to reduce the number of vertices, to reduce the space requirements, and to remove  noise (extraneous bends) from the polygon boundary (e.g., \cite{ArgeDMRT12,Douglas73,GuibasHMS93}). In this paper we consider a specific version of polygon simplification introduced by Aichholzer et al.~\cite{AichholzerHKPV14}, which keeps reflex corners intact, but minimizes the number of convex corners. Aichholzer et al. showed that such a simplification can help achieve faster solutions for many geometric problems such as answering shortest path queries (as shortest paths stay the same), computing Voronoi diagrams, and so on.

A simple polygon is a connected region without holes. Let $P$ be a polygon with $r$ reflex vertices and possibly with holes and islands. An \emph{island} is a simple polygon that lies entirely inside a hole. A \emph{reflex corner} of $P$ consists of three consecutive vertices $u,v,w$ on the boundary of $P$ such that the angle $\angle uvw$ inside $P$ is more than $180^\circ$. We refer the vertex $v$ as a \emph{reflex} vertex of $P$. The vertices  of $P$ that are not reflex are called \emph{convex} vertices. By a \emph{component} of $P$, we refer to a maximally connected region of $P$. A polygon \emph{$P'$ subsumes $P$}  if $P \subseteq P'$, each component $R'$ of $P'$ subsumes a distinct component $R$ of $P$, i.e., $R\subseteq R'$, and the reflex corners of $R$ coincide with the reflex corners of $R'$. A \emph{$k$-convex subsuming polygon} $P'$ contains at most $k$ convex vertices. A \emph{min-convex subsuming polygon} is a subsuming polygon that minimizes the number of convex vertices. Figure~\ref{f:subsume}(a) illustrates a polygon $P$, and Figures~\ref{f:subsume}(b) and (c) illustrate a subsuming polygon and a min-convex subsuming polygon of $P$, respectively. A \emph{subsuming chain} of $P'$ is a minimal path on the boundary of $P'$ whose end edges coincide with a pair of edges of $P$, as shown in Figure~\ref{f:subsume}(b). 

\begin{figure}[pt]
    \centering
    \includegraphics[width=0.9\textwidth]{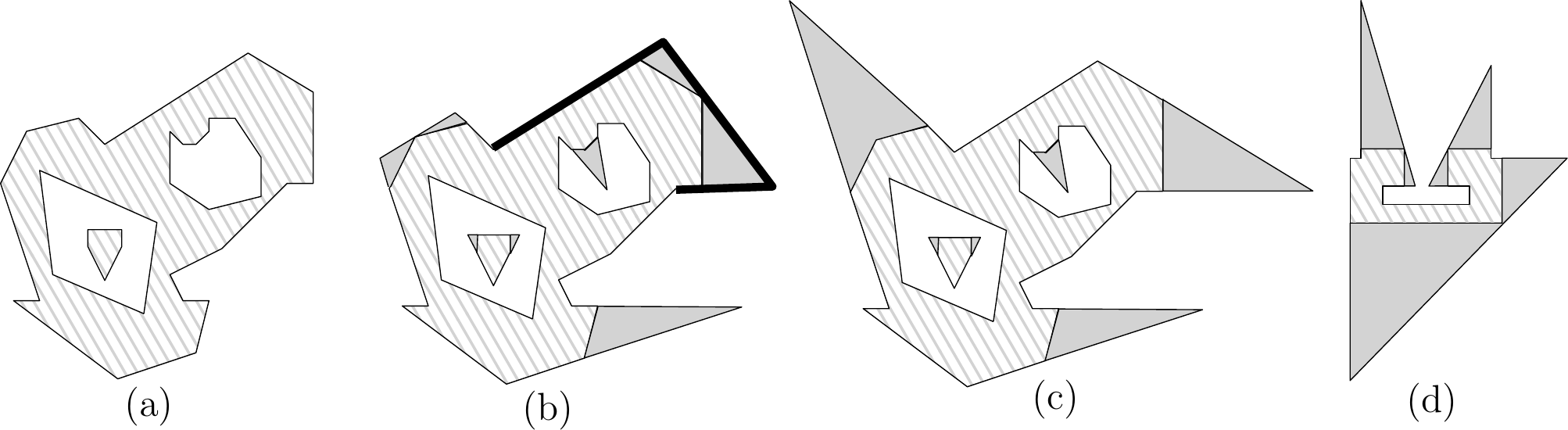}
    \caption{(a) A polygon $P$, where the polygon is filled and the holes are empty regions.
	(b) A subsuming polygon $P'$, where  $P'$ is the union of the filled regions. A 
	subsuming chain is shown in bold. 
	(c) A min-convex subsuming polygon $P'_{min}$, where $\A(P'_{min}) = \A(P)$. (d) A polygon $P$ 
	such that for any min-convex subsuming polygon $P'_{min}$, $\A(P)\not = \A(P'_{min})$. }
    \label{f:subsume}
\end{figure}

 Aichholzer et al.~\cite{AichholzerHKPV14} showed that for every polygon $P$ with $n$ vertices,
 $r>0$ of which are reflex, one can compute in linear time a subsuming polygon $P'$ 
 with at most  $O(r)$ vertices. Note that although a subsuming polygon with $O(r)$
 vertices always exists, no polynomial-time  algorithm is known for computing a
 min-convex subsuming polygon. Finding an optimal subsuming polygon seems challenging since
 it does not always lie on the arrangement of lines $\A(P)$ (resp., $\Av(P)$) determined by the edges (resp.,
 pairs of vertices) of the input polygon. Figure~\ref{f:subsume}(c)
 illustrates an optimal polygon $P'_{min}$ for the polygon $P$ of Figure~\ref{f:subsume}(a), 
 where $\A(P'_{min})=\A(P)$. On the other hand, Figure~\ref{f:subsume}(d) shows that
 a min-convex subsuming polygon  may not always lie on $\A(P)$ or $\Av(P)$. 
 Note that the input polygon of Figure~\ref{f:subsume}(d) 
 is a \emph{simple polygon}, i.e., it does not contain any hole. Hence determining min-convex 
 subsuming polygons seems challenging even for  simple polygons. In fact, 
 Aichholzer et al.~\cite{AichholzerHKPV14} posed an open question that asks to determine the
 complexity of computing min-convex subsuming polygons, where the input is restricted to simple polygons.  

\paragraph{Our contributions.} In this paper we show that the problem of computing a min-convex subsuming polygon is NP-hard for polygons  possibly with holes (Section~\ref{sec:negative}). We noted earlier that discretizing the solution space is a potential challenge, i.e., that the optimal  polygon may not always lie on the line arrangement determined by the input polygon (Figure~\ref{f:subsume}(d)). Interestingly, our NP-hardness result does not seem to utilize this challenge, instead, the hardness holds even when we restrict the subsuming chains to have constant length and to lie on $\A(P)$. 

A question that naturally appears in this context is whether such restrictions on subsuming chains can make the problem easier for nontrivial classes of polygons. For example, consider an $x$-monotone polygon\footnote{$P$ is $x$-monotone if every vertical line intersects $P$ at most twice.} $P$, e.g., see Figure~\ref{fig:m}(a). Then it is not difficult to see that there exists a min-convex polygon such that each subsuming chain has constant length and lies on $\A(P)$. The argument is simple except for the subsuming chains that covers the  two ends of $P$, e.g., see Figure~\ref{fig:m}(b). A simple proof is included in Section~\ref{sec:m}. 

 \begin{figure}[t]
    \centering
    \includegraphics[width=0.7\textwidth]{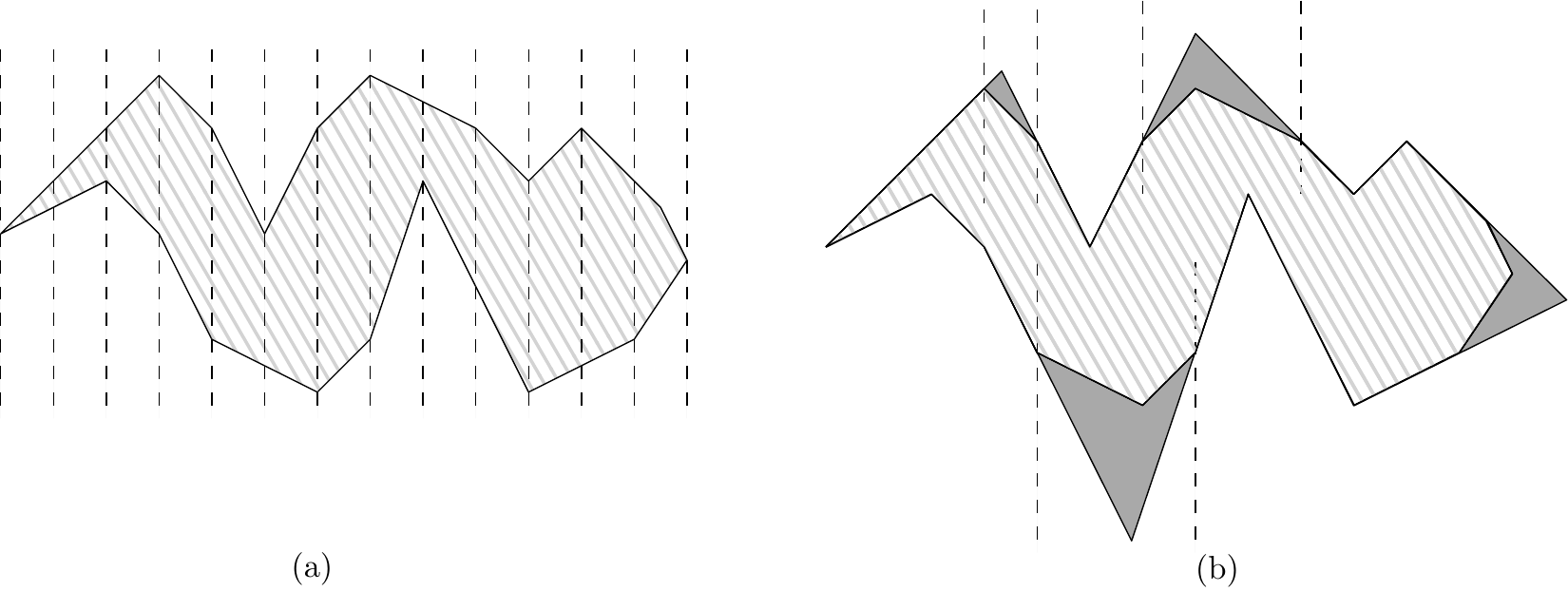}
    \caption{(a) An $x$-monotone polygon $P$. The dashed lines illustrate the monotonicity. (b) A min-convex subsuming polygon of $P$. The dashed lines illustrate the separation between successive chains.}
    \label{fig:m}
\end{figure}

We then show that the question can be answered affirmatively for arbitrary  simple polygons, i.e., for any simple polygon $P$, one can compute in polynomial time, a min-convex subsuming polygon $P_{min}$ under the restriction that the subsuming chains  are of constant length and lie on $\A(P)$. 

\paragraph{Organization.} The rest of the paper is organized as follows. Section~\ref{sec:negative} 
 presents the NP-hardness result for polygons with holes.  Section~\ref{sec:m} presents our observations on monotone polygons. Section~\ref{sec:positive}   describes the techniques for computing subsuming polygons for simple polygons. Finally, Section~\ref{sec:conclusion} concludes the  paper discussing directions to future research.

%%%%%%%%%%%% NEW SECTION %%%%%%%%%%%%%%%%%
\section{NP-hardness of Min-Convex Subsuming Polygon}
\label{sec:negative}
In this section we  prove that it is NP-hard to find a subsuming polygon 
 with minimum number of convex vertices. We denote the problem by \textsc{Min-Convex-Subsuming-Polygon}.
 We reduce the  NP-complete problem monotone planar 3-SAT~\cite{BK12}, 
 which is a variation of the 3-SAT problem as follows: Every clause  in a monotone planar 
 3-SAT consists of either three negated variables (\emph{negative clause}) or three non-negated 
 variables (\emph{positive clause}). Furthermore,  the bipartite graph constructed from the 
 variable-clause incidences, admits a planar drawing such that all the vertices corresponding 
 to the variables  lie along a horizontal straight line $l$, and all the vertices corresponding 
 to the positive (respectively, negative) clauses lie above (respectively, below) $l$. The 
 problem remains NP-hard even when each variable appears in at most four clauses~\cite{Darmann}. 

 \begin{figure}[t]
    \centering
    \includegraphics[width=0.9\textwidth]{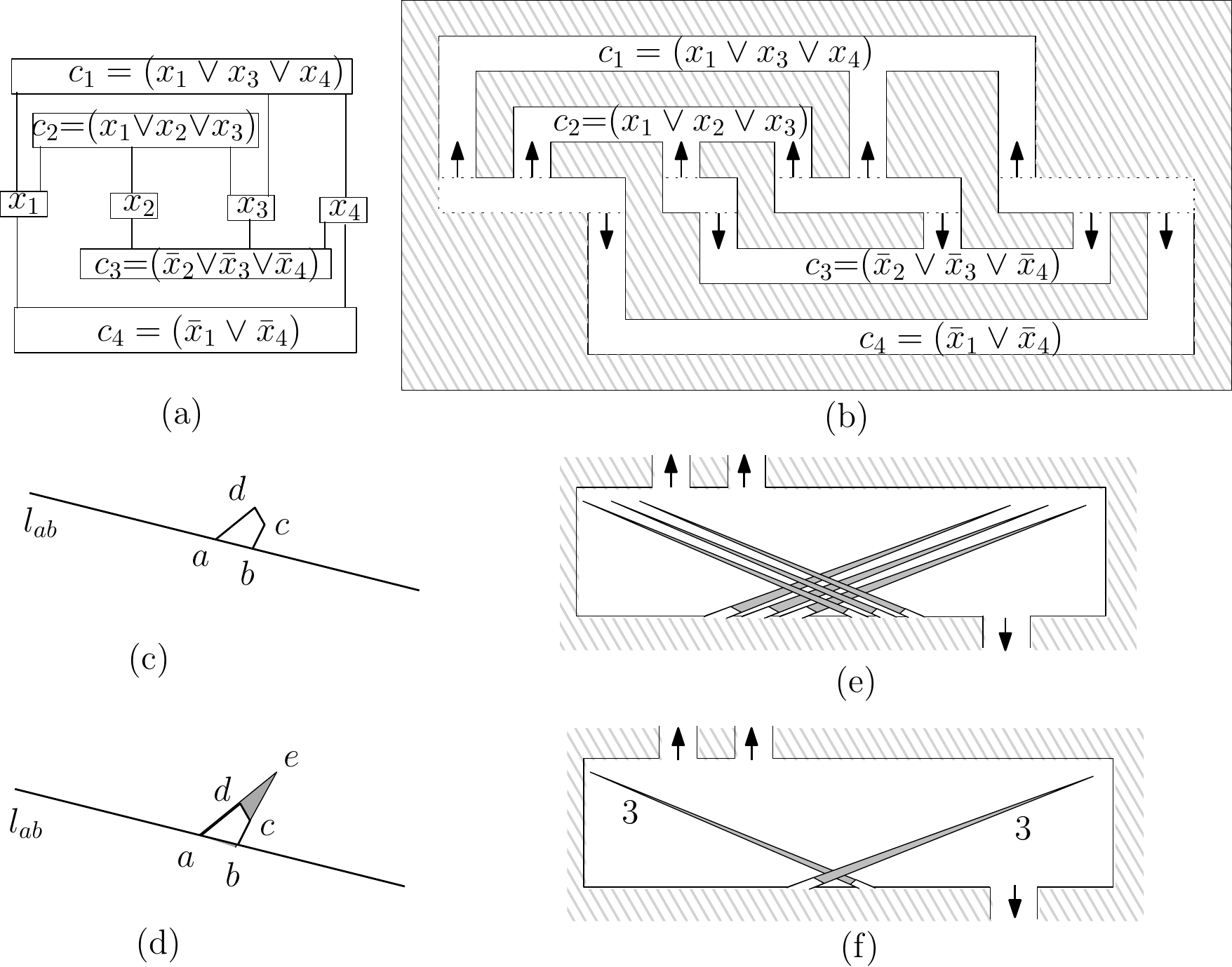}
    \caption{(a) An instance $I$ of monotone planar 3-SAT. (b) The 
    orthogonal polygon $P_o$ corresponding to $I$.
    (c)--(f) Illustration 
     for the variable gadget.}
    \label{f:nph}
\end{figure}
 
The idea of the reduction is as follows. Given an instance of a monotone  planar 
 3-SAT $I$ with variable set $X$ and clause set $C$, we  create a corresponding 
 instance $\mathcal{P}_I$ of \textsc{Min-Convex-Subsuming-Polygon}. 
 Let $\lambda$ be the number of convex vertices in $\mathcal{P}_I$. The reduction ensures that if 
 there exists a satisfying truth assignment of $I$, then $\mathcal{P}_I$ can be subsumed 
 by a polygon with at most $\lambda {-} |X||C|^2{-}3|C|$ convex vertices, and vice versa.

Given an instance $I$ of monotone planar 3-SAT, we first construct an orthogonal polygon $P_o$ with holes. 
 We denote each clause and variable using a distinct axis-aligned rectangle, which we refer to as the 
 \emph{c-rectangle} and \emph{v-rectangle}, respectively. Each edge connecting a 
 clause and a variable is  represented as a thin vertical strip, which we call an \emph{edge tunnel}.
 Figures~\ref{f:nph}(a) and (b) illustrate an instance of monotone planar 3-SAT and the corresponding
 orthogonal polygon, respectively. While adding the edge tunnels, we ensure for each v-rectangle that the tunnels coming from top lie to the left of all the tunnels coming from the bottom. Figure~\ref{f:nph}(b) marks the top and bottom edge tunnels by upward and downward rays, respectively. The v-rectangles, c-rectangles and the edge tunnels may form one or more holes, as it is shown by diagonal line pattern in Figure~\ref{f:nph}. We now transform $P_{o}$ to an instance $\mathcal{P}_I$ of \textsc{Min-Convex-Subsuming-Polygon}.

 We first introduce a few notations. 
 Let  $abcd$ be a convex quadrangle and let $l_{ab}$ be an infinite line that passes through $a$ and $b$. 
  Assume also that $l_{bc}$ and $l_{ad}$ intersect at some point $e$, and $c,d,e$ all lie on the 
  same side of $l_{ab}$, as shown in Figures~\ref{f:nph}(c)--(d).
  Then we call the quadrangle $abcd$ a \emph{tip} on $l_{ab}$, and the triangle $cde$ a \emph{cap} of $abcd$. The tip $abcd$ is called \emph{top-right}, if the slope of $ad$ is positive; otherwise, it is called a \emph{top-left} tip.
 
\paragraph{Variable gadget.} We construct variable gadgets from the v-rectangles. We add some top-right (and the same number of top-left) tips at the bottom side of the v-rectangle, as show in Figure~\ref{f:nph}(e). 
 There are three top-right and top-left tips in the figure. For convenience we show only one top-left
 and one top-right tip in the schematic representation, as shown in Figure~\ref{f:nph}(f).
 However,  we assign weight to these tips to denote how many tips there should be 
 in the exact construction. We will ensure a few more properties: (I) The caps do not intersect the boundary of 
 the v-rectangle, (II) no two top-left caps (or, top-right caps) intersect, and  
 (III) every top-left (resp., top-right) cap intersects all the top-right (resp., top-left) caps. 
 
Observe that each top-left tip contributes to two convex vertices such that covering them with a cap reduces the 
  number of convex vertices by 1. The peak of the cap reaches very close to the top-left corner of 
  the v-rectangle, which will later interfere with the clause gadget. Specifically, this cap will intersect
  any downward cap of the clause gadget coming through the top edge tunnels. Similarly, each 
  top-right tip contributes to two convex vertices, and the corresponding cap intersects any
  upward cap coming through the bottom edge tunnels.

Note that the optimal subsuming polygon $P$ cannot contain the caps from both the 
 top-left and top-right tips. We assign the tips with a weight of $|C|^2$.  
 In the hardness proof this will ensure that either the caps of top-right tips
 or the caps of top-left tips must exist in $P$, which will correspond
 to the true and false configurations, respectively.
 
\paragraph{Clause gadget.} Recall that, by definition, each clause consists of three variables and so it is incident to three edge tunnels. Figure~\ref{f:nphc}(a)  illustrates the transformation for a c-rectangle. 
 Here we describe the gadget for the positive clauses, and the construction for negative clauses is symmetric. 
 We add three downward tips incident to the top side of the c-rectangle, along its three edge tunnels.
 Each of these downward tip contributes to two convex vertices such that covering the tip with 
 a cap reduces the number of convex vertices by 1. Besides, the corresponding caps reach almost to the bottom 
 side of the v-rectangles, i.e., they would intersect the top-left caps of the $v$-rectangles. 
 Let these tips be $t_1,t_2,t_3$ from left to right, and let $\gamma_1,\gamma_2,\gamma_3$ be the corresponding caps.

 \begin{figure}[t]
    \centering
    \includegraphics[width=0.85\textwidth]{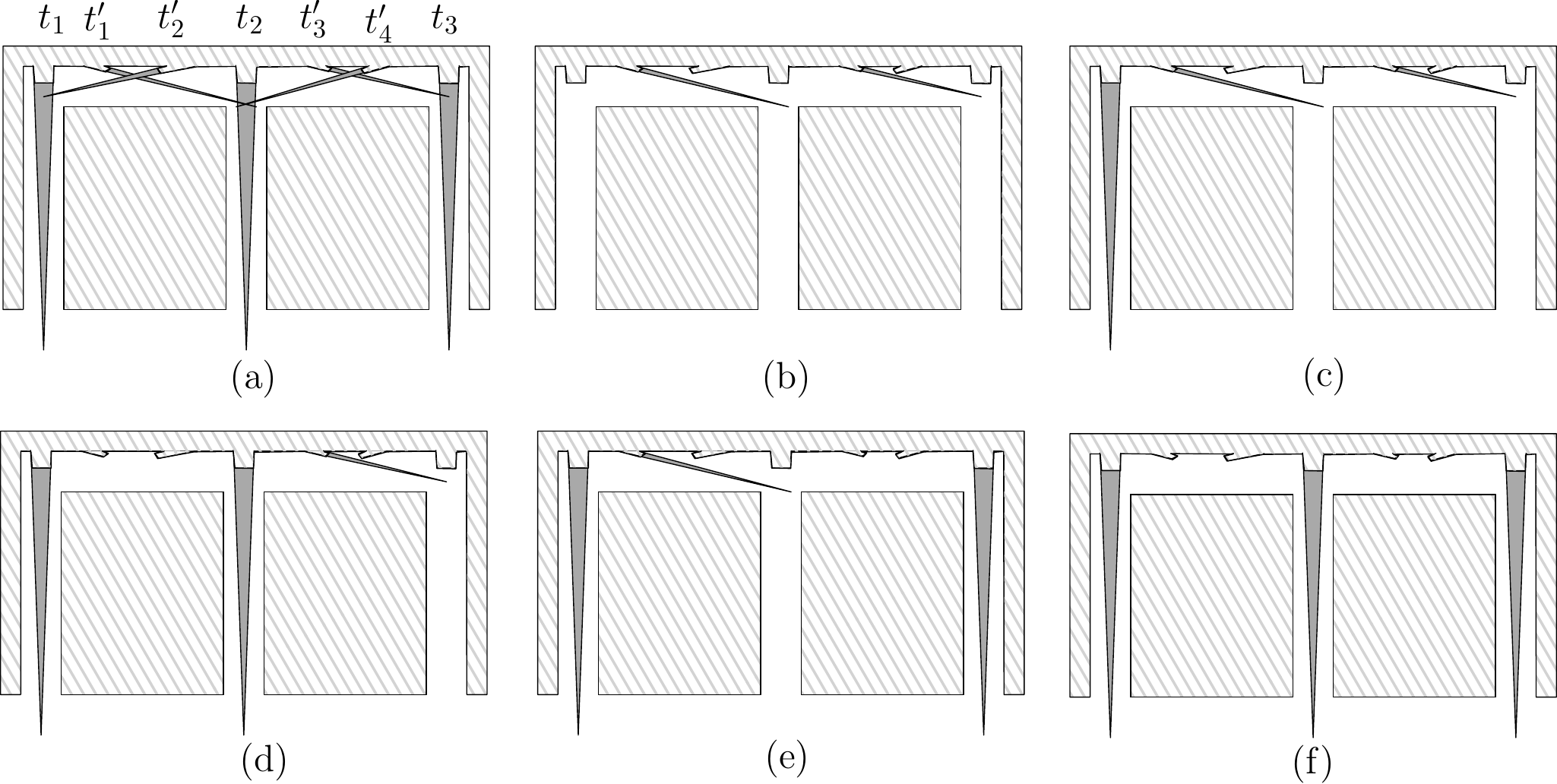}
    \caption{Illustration for the clause gadget.}
    \label{f:nphc}
\end{figure}
 
 We then add a down-left and a down-right tip at the top side of the $c$-rectangle between 
 $t_i$ and $t_{i+1}$, where $1\le i\le 2$, as shown in Figure~\ref{f:nphc}(a). Let the tips 
 be $t'_1,\ldots,t'_4$ from left to right, and let the corresponding caps be 
 $\gamma'_1,\ldots, \gamma'_4$. Note that the caps corresponding to $t'_j$ and  
 $t'_{j+1}$, where $j\in\{1,3\}$, intersect each other. Therefore, at most two of these four caps
 can exist at the same time in the solution polygon. Observe also that the caps
 corresponding to $t_1,t_2,t_3$ intersect the caps corresponding to $\{t'_2\},\{t'_1,t'_4\},
 \{t'_3\}$, respectively. Consequently, any optimal solution
 polygon containing none of $\{\gamma_1,\gamma_2,\gamma_3\}$ has at least 12 convex vertices along the top boundary of the c-rectangle, as shown in Figure~\ref{f:nphc}(b).
 
 We now show that any optimal solution polygon $P$ containing
 at least $\alpha>0$ caps from $\Gamma = \{\gamma_1,\gamma_2,\gamma_3\}$ have exactly 11 convex vertices 
 along the top boundary of the c-rectangle. We consider the following three cases:
 
 \textbf{Case 1 ($\alpha=1$):} If $\gamma_1$ (resp., $\gamma_3$) is in $P$,
  then $P$ must contain $\{\gamma'_1,\gamma'_3\}$ (resp., $\{\gamma'_2,\gamma'_4\}$). 
  Figure~\ref{f:nphc}(c) illustrates the case when $P$ contains $\gamma_1$.
  If $\gamma_2$  is in $P$, then $P$ must contain $\{\gamma'_2,\gamma'_3\}$.
  In all the above scenarios the number of convex vertices along the top boundary of the c-rectangle is 11.

  \textbf{Case 2 ($\alpha=2$):} If $P$ contains $\{\gamma_1,\gamma_3\}$, 
  then either $\gamma'_1$ or $\gamma'_4$ must be in $P$. Otherwise,
  $P$ contains either $\{\gamma_1,\gamma_2\}$ or $\{\gamma_2,\gamma_3\}$.
  If that $P$ contains 
  $\{\gamma_1,\gamma_2\}$, as in Figure~\ref{f:nphc}(d), then 
  $\gamma'_3$ must lie in $P$. In the remaining case, $\gamma'_2$ must lie in $P$.
  Therefore, also in this case the number of convex vertices along the top boundary of the c-rectangle  is 11.  

  \textbf{Case 3 ($\alpha=3$):} In this scenario $P$ cannot contain any of $\gamma'_1,
  \ldots, \gamma'_4$. Therefore, as shown in Figure~\ref{f:nphc}(e), the number of convex vertices along the top boundary of the c-rectangle  is 11.

As a consequence we obtain the following lemma.
\begin{lemma}
\label{lem:savings}
 If a clause is satisfied, then any optimal subsuming polygon  reduces exactly  
 three convex vertex from the corresponding c-rectangle.
\end{lemma}

\paragraph{Reduction.} Although we have already described the variable and clause gadgets, the optimal subsuming polygon still may come up with some unexpected optimization that interferes with the convex corner count in our hardness proof. 
 Figure~\ref{f:caution}(left) illustrates one such example. Therefore, we replace 
 each convex corner that does not correspond to the tips by a small polyline with 
 alternating convex and reflex corners, as shown   Figure~\ref{f:caution}(right).
 By construction, it is now straightforward to observe the following fact. 
 
\begin{myremark}\label{rmk2}
Let $r$ be a reflex vertex of $\mathcal{P}_I$, and let $r'$ be the first reflex vertex after $r$ while walking clockwise on the boundary of $\mathcal{P}_I$ starting at $r$. Then the number of convex vertices that can appear between  $r$ and $r'$ is at most two. Furthermore, if there are two convex vertices, then either they correspond to a tip, or form an   $180^\circ$ turn using only orthogonal line segments.   
\end{myremark}

 We now prove the NP-hardness of computing an optimal subsuming polygon.
 \begin{figure}[t]
    \centering
    \includegraphics[width=0.8\textwidth]{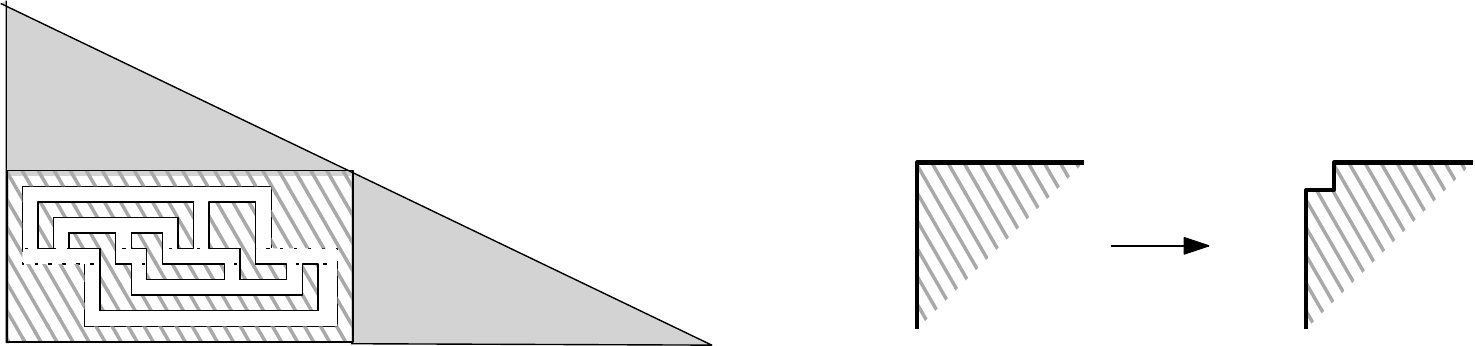}
    \caption{Refinement of $\mathcal{P}_I$.}
    \label{f:caution}
 \end{figure}

\begin{theorem}
\label{thm:hard}
 Finding an optimal subsuming polygon is NP-hard.
\end{theorem}
\begin{proof}
 Let $I=(X,C)$ be an instance of monotone planar 3-SAT and let $\mathcal{P}_I$ be the 
  corresponding instance of \textsc{Min-Convex-Subsuming-Polygon}. Let $\lambda$ be the number of convex vertices in $\mathcal{P}_I$. We now show that $I$ admits a satisfying truth assignment if and only if  $\mathcal{P}_I$ can be subsumed using a 
  polygon having at most $\lambda - |X||C|^2-3|C|$ convex vertices.
  
First assume that $I$ admits a satisfying truth assignment. 
 For each variable $x$, we choose  either the top-right caps or 
 the top-left caps depending on whether $x$ is assigned true or false.
 Consequently, we save at least $|X||C|^2$ convex vertices. 
 Consider any clause $c\in C$. Since $c$ is satisfied, one or more of its variables
 are assigned true. Therefore, for each positive (resp., negative) clause, we can
 have  one or more downward (resp., upward) caps that enter into the v-rectangles.
 By Lemma~\ref{lem:savings}, we can save at least three
 convex vertices from each c-rectangle. Therefore,
 we can find a subsuming polygon with  at most $\lambda - |X||C|^2-3|C|$ convex vertices.
 
Assume now that  some polygon $P$ with at most $\lambda - |X||C|^2-3|C|$ convex vertices
 can subsume $\mathcal{P}_I$. We now find a satisfying truth assignment for $I$. By Remark~\ref{rmk2} we can restrict our attention only to c- and v-rectangles. 
 Note that the maximum number of convex vertices that can be reduced from the c-rectangles 
 is at most $3|C|$. Therefore, $P$ must reduce at least $|C|^2$ convex vertices from
 each v-rectangle. Recall that in each v-rectangle, either the top-right or the top-left
 caps can be chosen in the solution, but not both. Therefore, the v-rectangles cannot
 help reducing more than $|X||C|^2$ convex vertices. If $P$ contains the top-right caps of the
 v-rectangle, then we set the corresponding variable to true, otherwise, we set it to false. 
 Since $P$ has at most $\lambda - |X||C|^2-3|C|$ convex vertices, and each c-rectangle
 can help to reduce at most $3$ convex vertices (Lemma~\ref{lem:savings}), $P$ must have at least
 one cap from $\gamma_1,\gamma_2,\gamma_3$ at each c-rectangle. Therefore, each clause must 
 be satisfied. Recall that the downward (resp., upward) caps coming from edge tunnels 
 are designed carefully to have conflict with
 the top-left (resp., top-right) caps of v-variables. Since top-left and top-right caps of v-variables
 are conflicting, the truth assignment of each variable  is consistent in  
 all the clauses  that contains it.
\end{proof}

It is straightforward from the construction of  $\mathcal{P}_I$ that no optimal subsuming polygon $P$  that subsumes  $\mathcal{P}_I$ can have a subsuming chain of length larger than 3, and there always exists  an optimal solution that lies on $\A(\mathcal{P}_I)$. Hence, Theorem~\ref{thm:hard} holds even under the  restriction that the subsuming chains must be of constant length and lie on $\A(\mathcal{P}_I)$.  
 In Sections~\ref{sec:m} and~\ref{sec:positive}, we will show that these restrictions make the problem polynomial-time solvable for simple polygons.

%%%%%%%%%%%%%%%%%%%%%%%%%%%%%%%%%%%%%%%
\section{Monotone Polygons}
\label{sec:m} 
In this section, we give a straightforward algorithm to compute a min-convex subsuming polygon of $x$-monotone polygons. In fact, we prove a stronger claim that every $x$-monotone polygon $P$ admits a min-convex subsuming polygon such that the subsuming chains are of constant length and lie on $\A(P)$.

Let $\up(P)$ and $\lo(P)$ denote the upper and lower chains of $P$, respectively. Moreover, let $u_1$ (resp., $u_m$) be the leftmost (resp., rightmost) vertex of $P$; notice that the vertices $u_1$ and $u_m$ are both convex (as $P$ is $x$-monotone) and are shared between $\up(P)$ and $\lo(P)$. Let $u_2,\dots,u_{m-1}$ (resp., $l_2,\dots,l_{m'-1}$) be the set of reflex vertices of $P$ that lie on $\up(P)$ (resp., on $\lo(P)$); we let $l_1=u_1$ and $l_{m'}=u_m$. For a reflex vertex $u_i$, where $2\leq i<m$, let $\ell^-(u_i)$ (resp., $\ell^+(u_i)$) denote the line  determined by the edge $(u_{i-1},u_i)$ (resp., $(u_{i},u_{i+1})$). Similarly, define the lines $\ell^-(l_i)$ and $\ell^+(l_i)$ for all $2\leq i<m'$. For $u_1$ and $u_m$, only $\ell^+(u_1)$ and $\ell^-(u_m)$ are defined. We next describe the algorithm.

First, consider $\up(P)$. Initially, let the simplified polygon $P'$ be $P$. For each reflex vertex $u_i$ on $\up(P)$, where $2\leq i<m-1$, consider the vertical slab defined by the two vertical lines through $u_i$ and $u_{i+1}$. If there is no convex vertex of $\up(P)$ in this slab, then the edge $u_iu_{i+1}$ must be an edge of any feasible solution. So, such an edge stays in $P'$. Otherwise, $\ell^+(u_i)$ and $\ell^-(u_{i+1})$ intersect each other at some point $p$ outside and above $P$ or on $P$. Then, we remove the chain of convex vertices between $u_i$ and $u_{i+1}$ and add the two edges $u_ip$ and $u_{i+1}p$ to $P'$; hence, reducing the number of convex vertices of $\up(P)$ between $u_i$ and $u_{i+1}$ to one. Next, we consider $\lo(P)$ and apply the same process to every two each vertex $l_i$ on $\lo(P)$, where $2\leq i<m'-1$.

It remains to show how to deal with the convex vertices that appear before the leftmost reflex vertex or after the rightmost reflex vertex on each chain. We show that these convex vertices can be reduced to at most two convex vertices, depending on the slopes of the edges incident to such \emph{reflex} vertices. In the following, as the second part of the algorithm, we discuss the details for convex vertices that appear before the leftmost reflex vertices; the convex vertices on the other end of the polygon can be handled similarly.

Consider $u_1$ (i.e., the leftmost vertex of $P$) and let $u_r$ and $l_r$ denote the leftmost reflex vertices on $\up(P)$ and $\lo(P)$, respectively. Let $\pi$ be the chain on the boundary of $P$ that connects $u_r$ to $l_r$ in counter-clockwise order (i.e., it contains $u_1$). To reduce the number of convex vertices on $\pi$, it is sufficient to check if $\ell^-(u_r)$ and $\ell^-(l_r)$ intersect at some point $p$ whose $x$-coordinate is less than that of both $u_r$ and $l_r$. For instance, if the slope of $\ell^-(u_r)$ is positive and the slope of $\ell^-(l_r)$ is negative, then $\ell^-(u_r)$ and $\ell^-(l_r)$ intersect at such point $p$; see Figure~\ref{fig:endvertices}(a). In this case, we can replace $\pi$ with two edges $u_rp$ and $l_rp$; hence, reducing the number of convex vertices on $\pi$ to one. Therefore, we can simplify $\pi$ as follows: if $\ell^-(u_r)$ and $\ell^-(l_r)$ intersect at such point $p$, then we replace $\pi$ with two edges $u_rp$ and $l_rp$ (hence, reducing the number of convex vertices on $\pi$ to one). Otherwise, if no such point $p$ exists (e.g., when the slope of $\ell^-(u_r)$ is negative, but the slope of $\ell^-(l_r)$ is positive), then both $\ell^-(u_r)$ and $\ell^-(l_r)$ must intersect the line passing through at least one of the edges incident to $u_1$. See Figure~\ref{fig:endvertices}(b-c). So, we can replace $\pi$ with three edges and reducing the convex vertices on $\pi$ to two. We perform a similar reduction on the path $\pi$ corresponding to the rightmost convex vertex of $P$ and its ``closest'' reflex vertices from each chain on the other end of $P$. Let $P^*$ be the resulting simplified polygon.

\begin{figure}[t]
\centering
\includegraphics[width=0.80\textwidth]{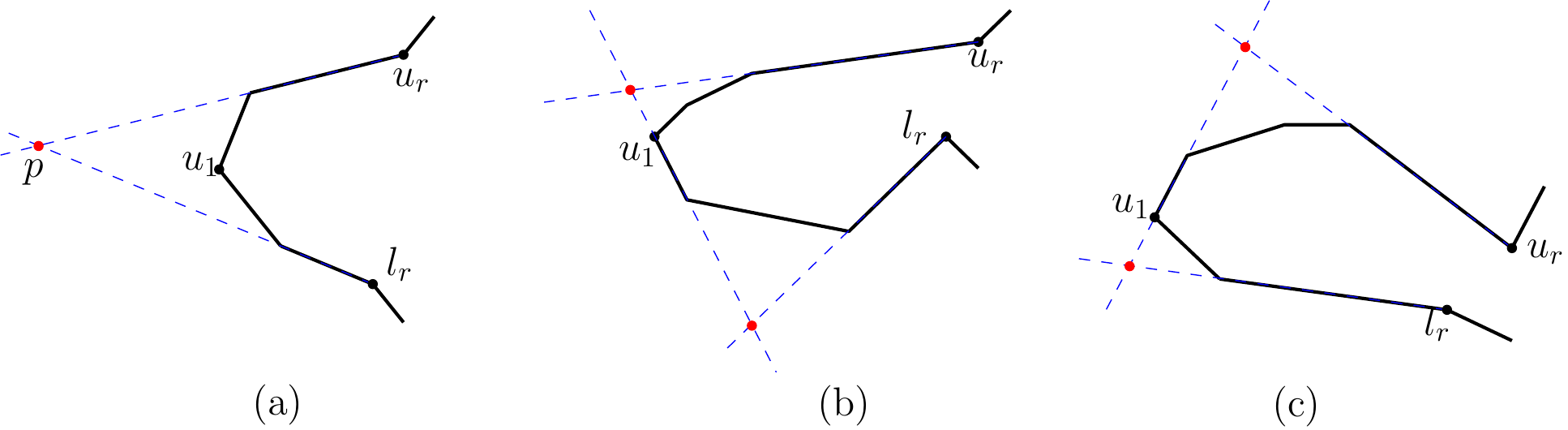}
\caption{How to reduce the convex vertices around the leftmost vertex $u_1$ of $P$.}
\label{fig:endvertices}
\end{figure}

To see the monotonicity of $P^*$, we note that in each slab considered in the first part of the algorithm, at most two new edges are introduced that lie inside the slab. Therefore, the edges from different slabs are disjoint. Moreover, the edges introduced in the second part of the algorithm (i.e., when dealing with the leftmost and rightmost convex vertices $u_1$ and $u_m$) do not violate the monotonicity of $P^*$.

For every two consecutive reflex vertices, $P^*$ has exactly one convex vertex (resp., has no convex vertex) between them if $P$ had at least one (resp., had none) between them. Since there must be at least one convex vertex between every two reflex vertices that had at least one convex vertex between them, any simplified polygon must have at least as many convex vertices as $P^*$ between every two consecutive reflex vertices. Moreover, one can easily verify that $P^*$ has the minimum number of convex vertices generated in the second part of the algorithm. Therefore, $P^*$ is optimal and we hence have the following result.
\begin{theorem}
Given a monotone polygon $P$ with $n$ vertices, a subsuming polygon of $P$ with the minimum number of convex vertices can be computed in $O(n \log n)$ time.
\end{theorem}

%%%%%%%%%%%% NEW SECTION
\section{Computing Subsuming Polygons}
\label{sec:positive}
In this section, we show that for any simple polygon $P$, a min-convex subsuming polygon $P_{min}$ can be computed in polynomial time under the restriction that 
 the subsuming chains  are of constant length and lie on $\A(P)$. % Therefore, if $t=O(1)$, then the time complexity of our algorithm is polynomial in $n$. 
  We first present definitions and preliminary results on outerstring graphs, which will be an important tool for computing subsuming polygons. 

\paragraph{Independent set in outerstring graphs.} A graph $G$ is   a \emph{string graph} if it is an intersection graph of a 
 set of simple curves in the plane, i.e., each vertex of $G$ is a mapped to a curve (string),
 and two vertices are adjacent in $G$ if and only if the corresponding curves intersect. 
 $G$ is an \emph{outerstring graph} if the underlying  curves lie interior to a simple 
 cycle $C$, where each curve intersects $C$ at one of its  endpoints. 
 Figure~\ref{f:positive}(a) illustrates an outerstring graph and the corresponding 
 arrangement of curves. Later in our algorithm, the polygon will correspond to 
 the cycle of an outerstring graph, and some polygonal chains attached to the 
 boundary of the polygon will correspond to the strings of that outerstring graph.

 A set of strings is called \emph{independent} if no two strings in the set intersect,
 the corresponding vertices in $G$ are called an independent set of vertices.  
 Let $G$ be a weighted outerstring graph with a set $\mathcal{T}$ of weighted strings. A \emph{maximum weight independent set} $\I(\mathcal{T})$  (resp., $\I(G)$) is a set of independent strings $T\subseteq \mathcal{T}$ (resp., vertices) that maximizes the sum of the weights of the strings in $T$. By $|\I(G)|$ we denote the weight of $\I(G)$.

Let $\Gamma(G)$ be the arrangement of curves that corresponds to $G$; e.g., see Figure~\ref{f:positive}(a). Let $R$ be a geometric representation of $\Gamma(G)$, where $C$ is represented as a simple polygon $P$, and each curve is represented as a simple polygonal chain inside $P$ such that one of its endpoints coincides with a distinct vertex of $P$. Keil et al.~\cite{Keil15} showed that given a geometric representation $R$ of $G$, one can compute a maximum weight independent set of $G$ in $O(s^3)$ time, where $s$ is the number of line segments in $R$.

\begin{theorem}[Keil et al.~\cite{Keil15}]
\label{thm:outerstring}
Let $G$ be a weighted outerstring graph.  Given a geometric representation $R$ of $G$, there exists a dynamic programming algorithm that computes a
 maximum weight independent set of $G$ in $O(s^3)$ time, where $s$ is the number of straight line segments in $R$.
\end{theorem}

 \begin{figure}[t]
    \centering
    \includegraphics[width=0.85\textwidth]{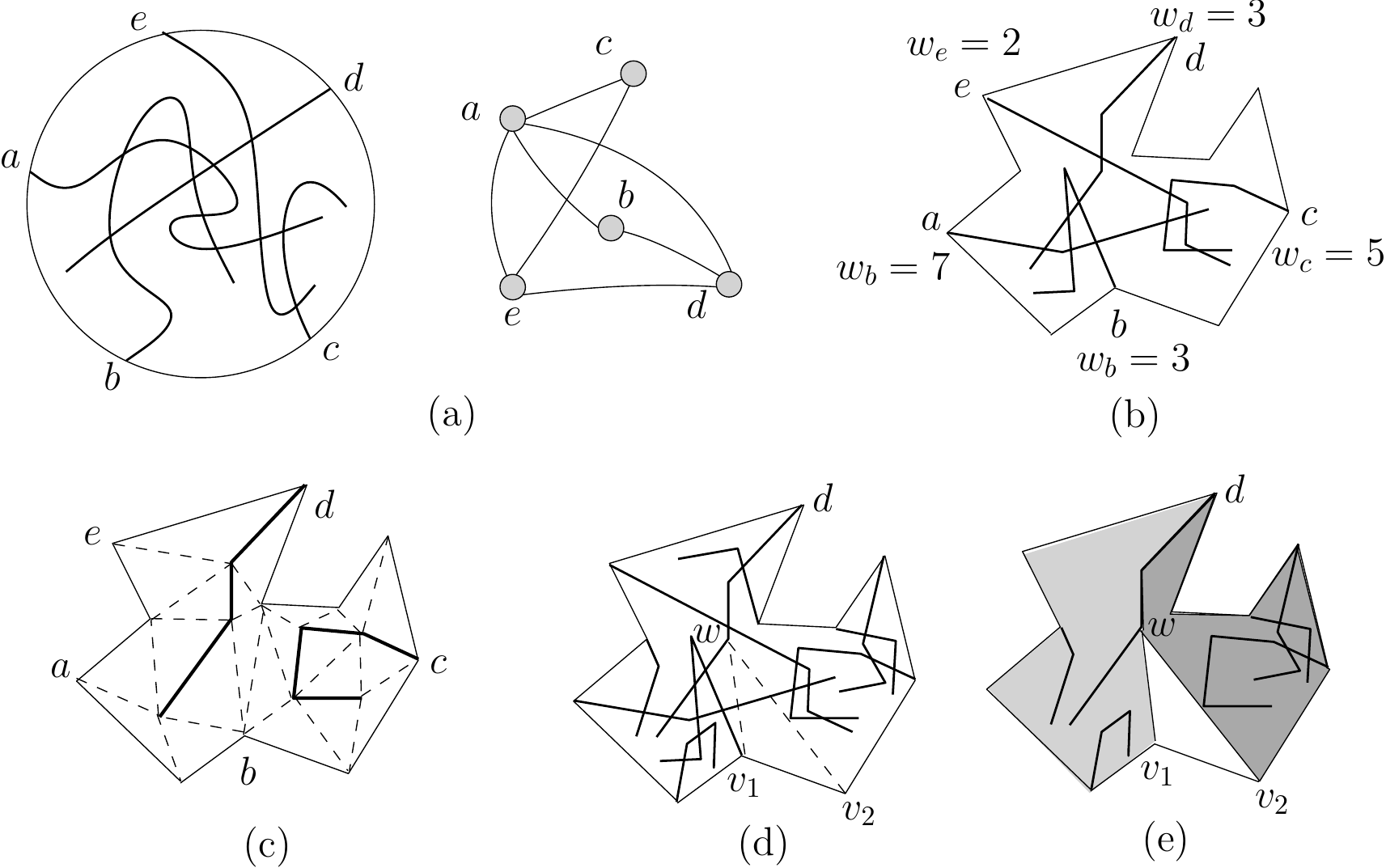}
    \caption{(a) Illustration for $G$ and $\Gamma(G)$. (b) A geometric representation $R$ of $G$.
	(c) A triangulated polygon obtained from an independent set of $G$.
	(d)--(e) Dynamic programming to find maximum weight independent set. }
    \label{f:positive}
\end{figure}

Figure~\ref{f:positive}(b) illustrates a geometric representation $R$ of some $G$,
 where each string is represented with at most 4 segments.  Keil et al.~\cite{Keil15} observed that any 
 maximum weight independent set of strings can be triangulated to create a triangulation $P_t$ of $P$, as shown in Figure~\ref{f:positive}(c). They used this idea to design a dynamic programming algorithm, as follows. Let $\mathcal{T}$ be the strings in $R$.
 Then the problem of finding $\I(\mathcal{T})$ can   be solved by dividing the problem into
  subproblems,  each described using only two points of $R$.
 We illustrate how the subproblems are computed very briefly using Figure~\ref{f:positive}(d). 
 Let $P(v_1,v_2)$ be the problem of finding  $\I(\mathcal{T}_{v_1,v_2})$, where  $\mathcal{T}_{v_1,v_2}$ consists of 
 the strings that lie to the left  of $v_1v_2$. Let $wv_1v_2$ be a triangle in $P_t$,
 where $w$ is a point on some string $d$ inside $P(v_1,v_2)$; see Figure~\ref{f:positive}(d). Since $P_t$ is a triangulation of the 
 maximum weight string set, $d$ must be a string in the optimal solution. Hence 
 $P(v_1,v_2)$ can be computed from the solution to the subproblems $P(v_1,w)$ and $P(w,v_2)$,
 as shown in Figure~\ref{f:positive}(e).
 Keil et al.~\cite{Keil15} showed that there are only a few different cases depending 
 on whether the points describing the subproblems belong to the polygon or 
 the strings.  We will use this idea of computing $\I(\mathcal{T})$ to compute subsuming polygons.

\paragraph{Subsuming polygons via outerstring graphs.} Let $P=(v_0,v_1,\ldots,v_{n-1})$ be a simple polygon with $n$ vertices, $r>0$ of which are reflex vertices. A \emph{convex chain} of  $P$ is a path $C_{ij} = (v_i,v_{i+1},\ldots,v_{j-1},v_j)$ of strictly convex vertices, where the indices are considered modulo $n$.

Let $P' = (w_0, w_1,\ldots, w_{m-1})$ be a subsuming polygon of $P$, where
  $\A(P') = \A(P)$, and the subsuming chains are of length at most $t$. Here, by ``length'', we mean the number of edges (not the Euclidean length).
  Let  $C'_{qr} = (w_q,\ldots, w_r)$ be  a subsuming chain of $P'$. 
  Then by definition, there is a corresponding convex chain $C_{ij}$ in $P$
  such that the edges $(v_i,v_{i+1})$ and $(v_{j-1},v_{j})$ coincide with
  the edges  $(w_q,w_{q+1})$ and $(w_{r-1},w_{r})$. 
  We call the vertex $v_{i}$ the \emph{left support} of $C'_{qr}$.         % $w_{q-1}$. 
  Since $\A(P') = \A(P)$, the chain $C'_{qr}$  must lie on $\A(P)$. 
  Moreover, since $P'$ is a
  min-convex subsuming polygon, the number of vertices in $C'_{qr}$
  would be at most the number of vertices in $C_{ij}$.
  
We  claim that the number of paths in $\A(P)$ from $v_i$ to $v_j$ is at most $n^t$, where $t=O(1)$  is an upper bound on the length of the subsuming chains. Thus any subsuming chain can have at most  $(t-1)$ line segments. Since there are only $O(n)$ straight lines in the arrangement $\A(P)$, there can be at most $n^j$ paths of $j$ edges, where $1\le j \le t-1$. Consequently, the number of candidate  chains that can subsume $C_{ij}$ is  $O(n^{t})$.

\begin{lemma}
Given a simple polygon $P$ with $n$ vertices, 
 every convex chain $C$ of $P$ has at most $O(n^t)$ candidate subsuming chains in $\A(P)$,
 each of length at most $t$.
\end{lemma}

 \begin{figure}[t]
    \centering
    \includegraphics[width=0.7\textwidth]{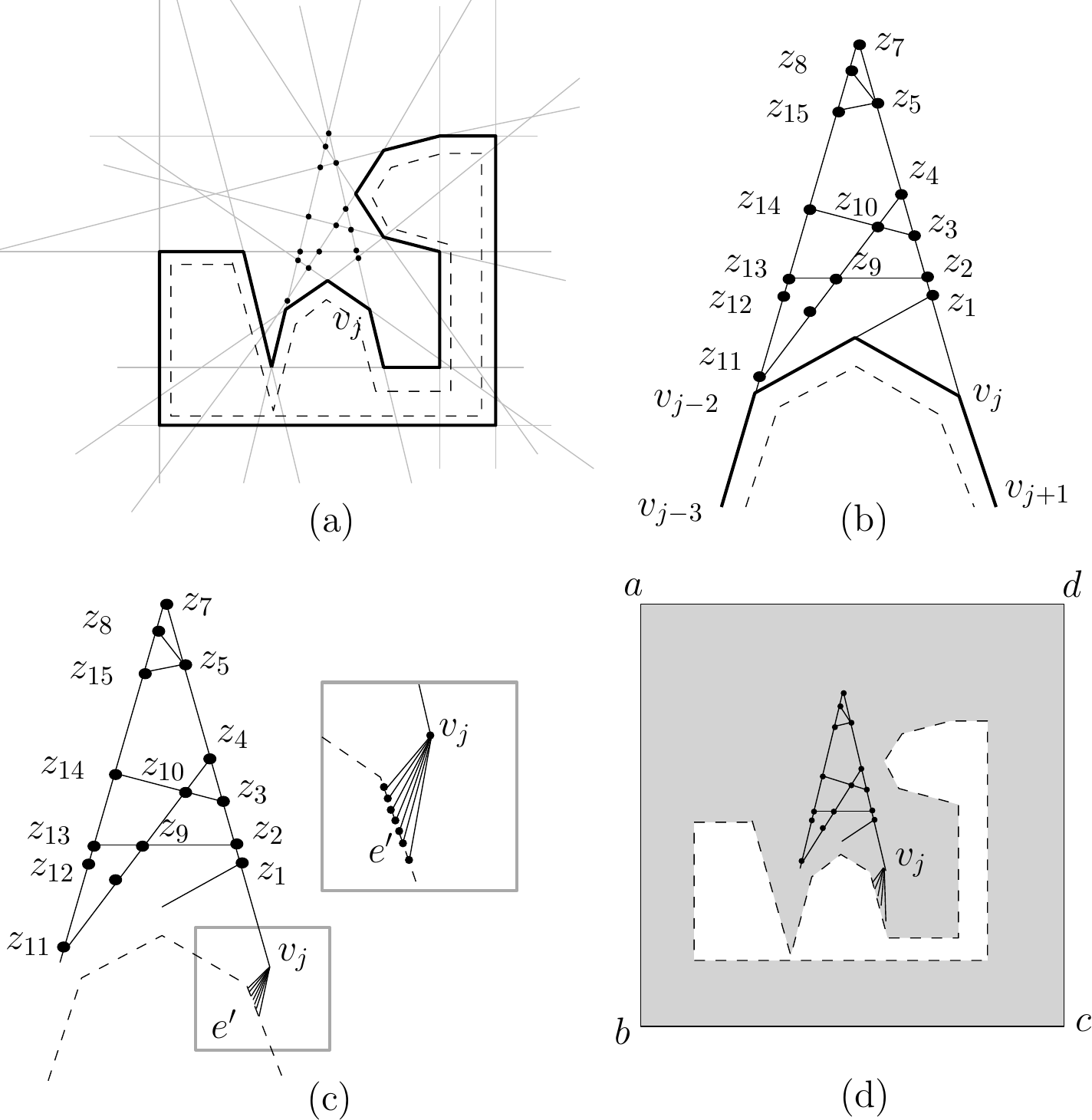}
    \caption{(a) Illustration for the polygon $P$ (in bold), $\A(P)$ (in  gray), and $Q$ (in dashed lines). (b) Chains of $v_j$. (c) Attaching the strings to $Q$. (d) Dynamic programming inside the gray region. }
    \label{f:subsuming}
\end{figure}

  In the following, we construct an outerstring graph using these candidate subsuming chains. We first compute a simple polygon $Q$ interior to $P$ such that  for each edge $e$ in $P$, there exists a corresponding edge $e'$ in $Q$ which is  parallel to $e$ and the perpendicular distance between $e$ and $e'$ is $\epsilon$, as shown in dashed line in Figure~\ref{f:subsuming}(a). We choose $\epsilon$ sufficiently small\footnote{Choose $\epsilon=\delta/3$, where $\delta$ is the distance between the closest visible pair of boundary points.} such that  for each component $w$ of $P$, $Q$ contains exactly one component inside $w$. We now construct the strings. Let $v_j$ be a convex corner of $P$. Let $S_j$ be the set of candidate subsuming chains such that for each chain in $S_j$, the left support of the chain appears before $v_j$ while traversing the unbounded face of $P$ in clockwise order.  For example, the subsuming chains that correspond to $v_j$  are $(v_{j-2},z_1,v_{j+1})$,
 $(v_{j-3},z_{13},z_2,v_{j+1})$,
 $(v_{j-3},z_{14},z_3,v_{j+1})$,
 $(v_{j-3},z_{11},z_4,v_{j+1})$,
 $(v_{j-3},z_{15},z_5,v_{j+1})$,
 $(v_{j-3},z_{8},z_5,v_{j+1})$,
 $(v_{j-3},z_7,v_{j+1})$,
 as shown in Figure~\ref{f:subsuming}(b). For each of these chains, we
 create a unique endpoint on the edge $e'$ of $Q$, where $e'$ corresponds to the edge $v_jv_{j+1}$ in $P$,
 as shown in Figure~\ref{f:subsuming}(c). We then attach these chains  to $Q$ by adding a segment from $v_j$
 to its unique endpoint on $Q$. 
 
 We attach the chains for all the convex vertices of $P$ to $Q$. Later we will use 
 these chains as the strings of an outerstring graph.
 We then assign each chain a weight, which is the number of convex vertices of $P$ it can reduce.
 For example in Figure~\ref{f:subsuming}(b), the weight of the chain $(v_{j-3},z_{8},z_5,v_{j+1})$  is one.

Although the strings are outside of the simple cycle, it is straightforward 
 to construct a representation with all the strings inside a simple cycle $Q$:  Consider placing a dummy vertex at the intersection points of the arrangement, and then find a straight-line  embedding  of the resulting planar graph such that the boundary of $Q$ corresponds to the outer face of the embedding. Consequently,  $Q$ and its associated strings correspond to an outerstring graph representation $R$. Let $G$ be the underlying outerstring graph. We now claim that any $\I(G)$ corresponds to a min-convex subsuming polygon of $P$.
\begin{lemma}
Let $P$ be a simple polygon, where there exists a min-convex subsuming 
 polygon that lies on $\A(P)$, and let $G$ be the corresponding  outerstring graph.
 Any maximum weight independent set of $G$
 yields a min-convex subsuming polygon of $P$.
\end{lemma}
\begin{proof}
Let $T$ be a set of strings that correspond to a maximum weight independent set of  $G$. Since $T$ is an independent set, the corresponding subsuming chains do not create edge crossings. Moreover, since each subsuming chain is weighted by the number of convex corners it can remove, the subsuming chains corresponding to $T$ can remove $|\I(G)|$ convex corners in total.

Assume now that there exists a min-convex subsuming polygon that can remove at least $k$
 convex corners. The corresponding subsuming chains would correspond to an independent 
 set $T'$ of strings in $G$. Since each string is weighted by the number of convex corners the corresponding subsuming chain can remove, the weight of $T'$ would be at least $k$.
\end{proof}

\paragraph{Time complexity.} To construct $G$, we first placed a dummy vertex at the intersection 
 points of  the chains, and then computed a straight-line  embedding of the resulting 
 planar graph such that all the vertices of $Q$ are on the outerface. Therefore, the geometric representation used at most 
 $nt$ edges to represent each  string. Since each convex vertex of $P$ is associated with at most $ O(n^t)$ strings,
 there are at most $n\times O(n^t)$ strings in $G$.  
 Consequently, the total number of segments used in the geometric 
 representation is $O(tn^{2+t})$. 
 A subtle point here is that the strings in our representation may partially overlap, and more than three strings may intersect at one point.
 Removing such degeneracy does not increase the asymptotic size of the representation.
 Finally, by Theorem~\ref{thm:outerstring}, one can compute
 the optimal subsuming polygon in $O(t^3n^{6+3t})$ time.

The complexity can be improved further as follows. Let $abcd$ be a rectangle that contains all the intersection points of $\A(P)$. Then, every optimal solution can be extended to a triangulation of the closed region between $abcd$ and $Q$. Figure~\ref{f:subsuming}(d) illustrates this region in gray. We can now apply dynamic programming similar to the one described above to compute the maximum weight independent string set, where each subproblem finds a maximum weight set inside some subpolygon. Each such subpolygon can be described using two points $v_1,v_2$, each lying either on $Q$ or on some string, and a subset of $\{a,b,c,d\}$ that helps   enclosing the subpolygon.
 
Since there are $n\times O(n^t)$ strings, each containing at most $t$ points, 
 the number of vertices that correspond to the strings is $O(tn^{1+t})$. We will refer these as the \emph{string vertices}. Note that the number of total vertices in the geometric representation is also $O(tn^{1+t})$. If the subproblem is bounded by two string vertices, or one string vertex and one polygon vertex, then  similar to Keil et al.~\cite{Keil15}, we  can use a pair of vertices to describe a subproblem. However, sometimes we need more information to describe a subproblem, e.g., assume that the subproblem is bounded  from one side by the point $a$ and some vertex $v$ (corresponding to a string), and from the other side by the point $d$ and some vertex $v'$ (corresponding to a string). For these problems, we need a subset of $\{a,b,c,d\}$ to describe the problem boundary. Therefore, we 
 define our dynamic programing table to be $\mathcal{D}[x,y,z]$, where $x$ and $y$ corresponds to the string or polygon vertices, and $z$ corresponds to the constant size additional description of the boundary (whenever needed). Thus the size of the dynamic programming table is $O(tn^{1+t})\times O(tn^{1+t})\times O(1)$. Since there are at most $O(tn^{1+t})$ string vertices, there can be at most  $O(tn^{1+t})$ candidate triangles $v_1v_2w$ (e.g., Figure~\ref{f:positive} (e)). Consequently, we can fill an entry of the table in $O(tn^{1+t})$ time. Hence the dynamic program takes at most $O(t^3n^{3+3t})$ time in total.
\begin{theorem}\label{thm:e}
Given a simple polygon $P$ with $n$ vertices, one can compute in polynomial time, a min-convex subsuming polygon under the restriction that the subsuming chains must be of constant length and lie on $\A(P)$.
\end{theorem}

\paragraph{Generalizations.} We can further generalize the results for any given line arrangements. However, such a generalization may increase the time complexity. For example, consider the case when the given line arrangement is $\Av(P)$, which is determined by the pairs of vertices of $P$. Since we now have $O(n^2)$ lines in the arrangement $\Av(P)$, the time complexity increases to  $O(t^3(n^2)^{3+3t})$, i.e.,  $O(t^3n^{6+6t})$.

%%%%%%%%%%%% NEW SECTION %%%%%%%%%%%%%%%%%
\section{Conclusion}
\label{sec:conclusion}
In this paper, we developed a polynomial-time algorithm that can compute 
 optimal subsuming polygons for a given simple polygon  in restricted settings.
 On the other hand, if the polygon contains holes, then we showed that the problem of computing an optimal subsuming polygon is NP-hard. Therefore, the question of whether the problem is polynomial-time solvable for simple polygons~\cite{AichholzerHKPV14}, remains open. Note that islands are crucial in our hardness proof. The complexity of the problem for polygons with holes (but without any island) is also open. Since the optimization in one hole is independent of the optimization in the other holes of the polygon, resolving the complexity for polygon with holes would readily give important insight about the complexity for simple polygon.

Our algorithm can find an optimal solution if the optimal subsuming polygon lies on some prescribed arrangement of lines, e.g., $\A(P)$ or $\Av(P)$. The running time of our algorithm depends on the length of the subsuming chains, i.e., the running time 
 is polynomial if the subsuming chains are of constant length. However, there exist polygons whose optimal subsuming polygons contain subsuming chains of length $\Omega(n)$. Figure~\ref{f:spiral} illustrates such an example optimal solution that is lying on $\A(P)$. An  interesting research direction would be to examine whether there exists a good approximation algorithm for the general problem.

 Recently, Lubiw et al.~\cite{DBLP:journals/corr/abs-1802-06699} showed that the problem of drawing a graph inside a polygonal region is hard for the existential theory of the reals. The  subsuming polygon problem can also be viewed as a constrained graph drawing problem whereas the subsuming chains are modeled by edges that need to be drawn outside the polygon, possibly with bends. The goal is to find a crossing-free drawing of these edges that minimizes the total number of  bends. It would be interesting to examine whether the problem is $\exists\mathbb{R}$-hard   
 in such a graph drawing model.
 
\begin{figure}[t]
\centering
\includegraphics[width=0.3\textwidth]{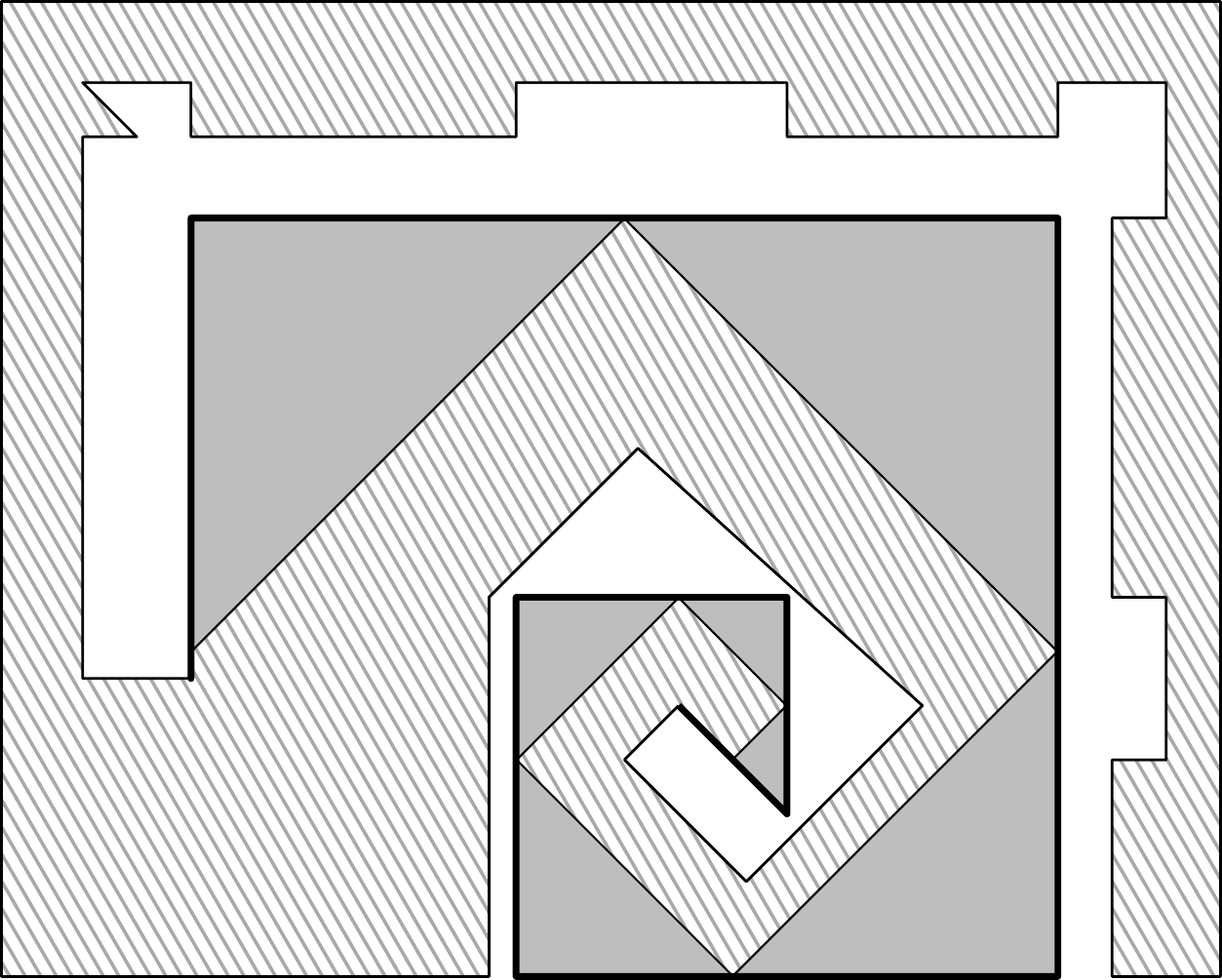}
\caption{Illustration for the case when the optimal subsuming polygon contains a subsuming chain of length $\Omega(n)$. 
The subsuming chain is shown in bold.}
\label{f:spiral}
\end{figure}

%%%%%%%%%%%% NEW SECTION %%%%%%%%%%%%%%%%%
\bibliographystyle{plain}
\bibliography{ref}

\end{document}